\documentclass[
 aip,
 jmp,
 amsmath,amssymb,
 reprint,
 onecolumn,
 11pt, a4paper,
 author-numerical,%
]{revtex4-1}

\usepackage{amsthm}
\usepackage{epsfig}
\usepackage{graphicx}
\usepackage{paralist}
\usepackage{enumerate}
\usepackage{colordvi}
\usepackage{tikz}

\newcommand{\be}{\begin{eqnarray}}
\newcommand{\ee}{\end{eqnarray}}

\newcommand{\bn}{\begin{eqnarray*}}       
\newcommand{\en}{\end{eqnarray*}}

\newcommand{\bea}{\begin{eqnarray*}}
\newcommand{\eea}{\end{eqnarray*}}
\newcommand{\ben}{\begin{eqnarray}}
\newcommand{\een}{\end{eqnarray}}
\newcommand{\beq}{\begin{equation}}
\newcommand{\eeq}{\end{equation}}

\newtheorem{theorem}{Theorem}
\newtheorem{lemma}{Lemma}
\newtheorem{proposition}{Proposition}

\newcommand{\C}{\ensuremath{\mathbb{C}}}

\newcommand{\R}{\ensuremath{\mathbb{R}}}

\newcommand{\la}{\langle}
\newcommand{\ra}{\rangle}

\newcommand{\D}{\mathrm{d}}

\newcommand{\LL}{\mathcal{L}}

\newcommand{\eu}{\mathrm{e}}

\newcommand{\eps}{\varepsilon}

\DeclareMathOperator*{\supp}{supp}
\begin{document}

\title{On the existence of bound states in asymmetric leaky wires}

\begin{abstract}

\end{abstract}

\author{Pavel Exner}%
\email{exner@ujf.cas.cz}
\affiliation{Department of Theoretical Physics, Nuclear Physics Institute, Czech Academy of Sciences, CZ-25068 \v{R}e\v{z} near Prague}
\affiliation{Doppler Institute, Faculty of Nuclear Sciences and Physical Engineering, Czech Technical University, B\v{r}ehov\'{a} 7, CZ-11519 Prague~1}%
\author{Semjon Vugalter}%
\email{semjon.wugalter@kit.edu}
\affiliation{Institute for Analysis, Karlsruhe Institute for Technology (KIT),\\
Kaiserstrasse 89, D-76133 Karlsruhe}%


\begin{abstract}
\noindent We analyze spectral properties of a leaky wire model with a potential bias. It describes a two-dimensional quantum particle exposed to a potential consisting of two parts. One is an attractive $\delta$-interaction supported by a non-straight, piecewise smooth curve $\LL$ dividing the plane into two regions of which one, the `interior', is convex. The other interaction component is a constant positive potential $V_0$ in one of the regions. We show that in the critical case, $V_0=\alpha^2$, the discrete spectrum is non-void if and only if the bias is supported in the interior. We also analyze the non-critical situations, in particular, we show that in the subcritical case, $V_0<\alpha^2$, the system may have any finite number of bound states provided the angle between the asymptotes of $\LL$ is small enough.
\end{abstract}

\maketitle

\section{Introduction}

\noindent There are various ways in which quantum particles can be guided and their analysis is without any doubt important both theoretically and practically. One of the deep results here is the connection between the geometry of the guiding and the existence of localized states. This became an object of intense interest since the end of the eighties\cite{ES89, GJ92, DE95, RB95}, even if in a particular case the effect was noted more than two decades earlier\cite{He65}. While the paradigmatic result refers to smoothly bent channels with hard, i.e. Dirichlet walls, the effect is surprisingly robust being present for Robin, i.e. mixed type boundaries\cite{BMT12} as well as for sharply broken Dirichlet channels\cite{He65,CLM92}. What is more, a similar behavior was observed in systems governed by other equations, Maxwell or fluid dynamics\cite{ELV94}; in the former case it was confirmed by a simple experiment\cite{CLM92}.

The binding due to bends persists even if the confinement to the guiding channel is `softer' being realized by a protracted potential `ditch'. A simple model in which the guiding potential is of the $\delta$-type, being supported by a curve, is usually referred to as a \emph{leaky quantum wire}\cite{Ex08}. The name is derived from the use of such operators to describe systems of semiconductor wires, the `leakiness' means that in contrast to the usual quantum graph models\cite{BK} they do not exclude quantum tunneling between different parts of such a wire. This makes them a more realistic model of actual nanowires the boundaries are which are potential jumps, often high but finite. A bent leaky wire also exhibits bound states\cite{EI01} the number and binding energies of which depend on its geometry.

The mechanism producing those bound states, as first indicated in Ref.~\onlinecite{GJ92}, is related to the singularity at the bottom of the spectrum of a straight channel which may produce an isolated eigenvalue pole as a result of a geometric deformation. An important element in the universality of the result, however, is the symmetry of the unperturbed guide. If it is violated the picture changes, as first observed by Dittrich and K\v{r}\'{\i}\v{z}\cite{DK02} who noted that in a planar guide whose one boundary is Dirichlet and the other Neumann the existence of bound states depends on which way we bend it.

Our aim in the present paper is to investigate a similar problem for leaky wires with a potential bias, that is, to perform a spectral analysis of operators of the type
\beq \label{formalHam}
H = -\Delta + V(x) - \alpha \delta(x-\LL)\,, \quad \alpha>0\,,
\eeq
in $L^2(\R^2)$, where the $\delta$-potential, sometimes also written as $-\alpha \delta(\cdot-\LL)\la \cdot,\delta(\cdot-\LL) \ra$, is supported by an infinite, piecewise smooth curve $\LL$ dividing the plane into two regions. We will be interested in the situation where the potential is constant, positive, and supported in one of those regions. We are going to demonstrate that, similarly to the mentioned result of Dittrich and K\v{r}\'{\i}\v{z} and analogous effects in other asymmetric guides --- cf. Ref.~\onlinecite{BMT12} and references therein --- the existence of the bound states depends of the way in which $\LL$ is bent with respect to the bias.

We have to add a caveat, however, concerning various analogies one could think of. Let us recall that an orientation-dependent binding may be observed also in `one-sided' problems, an example being represented by the Laplacian in an infinite planar region with an `attractive' Robin boundary\cite{EM14}. While similar at a glance, these system behave in fact in the opposite way. The waveguide of Ref.~\onlinecite{DK02} as well as our leaky wire Hamiltonian \eqref{formalHam} exhibit bound states if the Dirichlet boundary or the positive-potential region are bent \emph{outward}, while in the one-sided situation an outward bend of the Robin boundary is exactly the situation when the system does \emph{not} have bound states\cite{EM14}. A heuristic way to understand the difference is to observe the transverse behavior of the threshold-energy solution in the case of a straight boundary. In the first case it is tilted away from the Dirichlet boundary, for the operator \eqref{formalHam} it is constant in the potential-free region and decays exponentially in its counterpart. In contrast, in Robin domains the solution is localized in the vicinity of the boundary.

Let us describe briefly the contents of the paper. As a preliminary, we analyze in the next section the transverse part of the operator  \eqref{formalHam} in the situation when $\LL$ is a straight line. Next, in Sec.~\ref{s:Ham}, we introduce the object of our interest properly, list the assumptions and identify the essential spectrum. The main results are presented in Sec.~\ref{s:ext}--\ref{s:finite}. The first two of them deal respectively with situations where the potential bias is supported in the exterior (concave) and the interior (convex) one of the two regions to which the curve $\LL$ divides the plane. In both cases we start from a discussion of the example in which $\LL$ is a broken line and subsequently generalize the results to a more general class of curves. Our primary attention is paid to the critical case, $V_0=\alpha^2$, in which we demonstrate  that the discrete spectrum is non-void if and only if the bias is supported in the interior. We also discuss the existence of bound states in the non-critical situations. In particular, we show that the system may have any finite number of bound states for $V_0<\alpha^2$ provided the angle between the asymptotes of $\LL$ is small enough. On the other hand, in Sec.~\ref{s:finite} we will show that the discrete spectrum of the model is always finite.

\section{A preliminary: the spectrum of a one-dimensional operator with $\delta$-interaction and a potential jump}
\setcounter{equation}{0} \label{s:prelim}

\noindent If the guiding curve is a straight line the problem can be solved by separation of variables. It is useful to inspect first how the transverse part looks like in this case. Consider thus the operator
\beq
h = -\frac{\D^2}{\D x^2} -\alpha \delta (x) + V(x)\,,
\eeq
where $V(x) = V_0$ for $x > 0$ and $V(x) = 0$ otherwise. As usual with the point interactions\cite{AGHH} the $\delta$-potential enters the operator domain description through the boundary conditions matching the functions at the point $x=0$. Alternatively, one can define $h$ as the self-adjoint operator associated with the form
\beq
\phi \mapsto \|\phi'\|^2 -\alpha|\phi(0)|^2  + \la V\phi, \phi\ra
\eeq
defined on $H^1(\R)$. Elementary properties of $h$ are easy to be found.

\begin{lemma}
 \label{propChi}
  \begin{enumerate}[(i)]
      \item $\sigma_\mathrm{ess} (h) = [0, \infty )$.
			\item The operator $h$ has no eigenvalues for $V_0 \ge \alpha^2$.
			\item The operator $h$ has a unique eigenvalue $\mu = - \left( \frac{ \alpha^2 - V_0}{2\alpha}\right)^2$ for $V_0 < \alpha^2$.
			\item If $V_0 = \alpha^2$ the equation $h\psi = 0$ has a bounded weak solution $\psi \notin L^2(\R^1)$.
	\end{enumerate}
\end{lemma}
\begin{proof}
The prof of claims (i)--(iv) is straightforward. The solution mentioned in (iv) is given explicitly by $\psi (x) = 1$ for $x < 0$ and $\psi (x) = \eu^{-\sqrt {V_0} x} $ for $x > 0$.
\end{proof}

The case (d) will be naturally called \emph{critical}, and consequently, we shall use the terms \emph{supercritical} and \emph{subcritical} for cases (ii) and (iii), respectively.

\begin{lemma}
 \label{propFunc}
For $V_0 > 0$ and any $\varphi \in \C^2(\R_+) \bigcap L^2(\R_+)$ we have
\beq
\label{eq1}
\int_0^\infty ({|\varphi^{\prime}}|^2 + V_0|\varphi|^2)(x)\, \D x \ge \sqrt{V_0}\, |\varphi(0)|^2.
\eeq
\end{lemma}
\begin{proof}
The Euler equation for the functional on the right-hand side of \eqref{eq1} has a unique solution, namely
$$
\varphi (x) = \varphi (0)\, \eu^{-\sqrt{V_0}x}\,.
$$
\end{proof}

\section{A two-dimensional leaky-wire Hamiltonian with a potential jump}
\setcounter{equation}{0} \label{s:Ham}

\subsection{The Hamiltonian}
\label{propL}

\noindent As indicated in the introduction, the object of our interest is the singular Schr\"odinger operator
\beq \label{Ham}
H = -\Delta + V(x) - \alpha \delta_\LL
\eeq
in $L^2(\R^2)$, where $\LL:\:\R \to \R^2$ is an infinite planar curve without self-intersections, $\delta_\LL$ is the $\delta$-potential supported by $\LL$, and $V(x) \ge 0$. We add other assumptions, namely
  \begin{enumerate}[(a)]
  \item $L$ divides $\R^2$ into two regions such that one of them is convex. The trivial case of two halfplanes is excluded.
  \item $L$ consists of a finite number of a $\C^2$ segments.
  \item The natural (arc-length) parametrization of $L$ is used in the following.
  \item Asymptotes of $L(s)$ for $s\to \pm\infty$ exist and they are not parallel. Everywhere in this paper we will assume that in the polar coordinates the asymptotes coincide with the radial halflines of angles $\varphi = \beta$ and $\varphi = -\beta .$
	\end{enumerate}
Under these assumptions the quadratic form

\beq
\psi \mapsto \|\nabla\psi\|^2 + \la\psi, V\psi\ra -\alpha \int_\R |\psi(\LL(s))|^2\, \D s
\eeq
defined on $H^1(\R^2)$ is closed and below bounded and the unique self-adjoint operator associated with it is $H\equiv H_{\alpha,\LL,V}$ of \eqref{Ham} above. With respect to the potential we focus our attention on a particular case. By hypothesis (a) above the curve splits the plane into two regions, we assume that
	\label{propV}
  \begin{enumerate}
  \item[(e)] $\,V(x)=V_0 >0$ in one of these regions and $V(x)= 0$ in the other.
	\end{enumerate}

\subsection{The essential spectrum}

\noindent As usual the essential spectrum is determined by the behavior of the potential, both its regular and singular components, at large distances. In view of the assumption (d) we expect that asymptotically a separation of variables will play role, and consequently, Lemma~\ref{propChi} could be used. Indeed, we have the following result.
\begin{theorem} \label{main1}
\emph{(Location of the essential spectrum)} Under the assumptions (a)--(e) we have $\sigma_\mathrm{ess} (H) = [\mu,\infty)\,$, where $\mu = -\frac14 \alpha ^{-2}(\alpha ^2 -V_0)^2$ for $V_0 < \alpha ^2$ and $\mu = 0$ otherwise.
\end{theorem}

\noindent The claim of the theorem is obviously equivalent to a pair of implications expressed by the following two lemmata.

\begin {lemma} \label{HVZ1}
$\,\nu \in \sigma_\mathrm{ess}(H)$ holds for any $\nu \ge \mu $.
\end{lemma}
 \begin{proof}
Let $\nu \ge \mu$ be fixed. To prove the lemma it suffices to show that for any fixed $\varepsilon > 0$ one can find an infinite-dimensional subspace $M$ such that $\|(H-\nu) \psi \| < \varepsilon$ holds
for any $\psi \in M$. Denote $\zeta :=\mu-\nu \ge 0 $ and let $f \in \C_0^2(\R)$ be such that $\|f\|=1$ and $\big\| (-\frac{\D^2}{\D x^2} - \zeta)f\big\| < \Red{\frac14 \varepsilon}$. Furthermore, let $g \in \C_0^2(\R)$ be such that $\|g\|=1$ and $\| (h - \mu)g\| < \Red{\frac14 \varepsilon}$. The functions $f,g$ can always be found in view of the the fact that the essential spectrum of $-\frac{\D^2}{\D x^2}$ is $[0,\infty)$ and of Lemma~\ref{propChi}(iii). Next we choose a Cartesian system of coordinates such that the $x$-axis coincides with one of the asymptotes and the origin is at the point of the asymptotes crossing, and we put $ \psi_i (x,y) = f(x-a_i)g(y)$, where the numbers $a_i,\: i=1,2,\dots$, will be chosen later. Notice that if the $a_i$'s are sufficiently large the other asymptote, as well as the branch of the curve approaching it do not intersect with the support of $\psi_i$. Hence we have
\beq \label{HVZeq1}
\|(H-\nu)\psi_i\| \leq \Big\|(-\frac{\D^2}{\D x^2}-\zeta)f\Big\|  +\|(h-\mu)g\| + \Red{V_0 \|\psi_i\|_{\mathcal{A}_i}} +\|(\delta_{y=0}-\delta_\LL) \psi_i\|\,,
\eeq
\Red{where $\mathcal{A}_i$ is the part of the region between $\LL$ and the appropriate asymptote which lies in the support of $\psi_i$ and the last term is understood as the $L^2$-norm over the two curve segments contained in the border of $\mathcal{A}_i$.} The first two terms on the right-hand side of \eqref{HVZ1} can be estimated by \Red{$\frac12 \varepsilon$. Furthermore, in view of the assumptions (a), (d) above the other two other terms tend to zero if $f,g$ are kept fixed and $a_i \to \infty$.} This yields $\|(H-\nu)\psi_i\| < \varepsilon$ for all $a_i$ large enough. In addition, one can always choose the $a_i$'s in such a way that $\supp \{\psi_i\} \cap \supp \{\psi_j\} = \emptyset$ holds for $i \neq j $, hence Weyl's criterion applies.
\end{proof}

This result has to be complemented by the following claim, which we are going to prove by a method based on the partition of unity in the configuration space. This technique was originally developed for analysis of multi-particle Schr\"odinger operators by G.~Zhislin \cite{Zh60}.

\begin {lemma} \label{HVZ2}
 $\:\sigma _\mathrm{ess} (H) \cap (-\infty, \mu) = \emptyset$.
\end{lemma}
\begin{proof}
Our aim is to show that for any fixed $\varepsilon >0$ one can find a finite-dimensional subspace $M$ such that for all $\psi \in \C^2_0 (\R^2)$ with $\psi \perp M$ we have the inequality
$$
\la H \psi, \psi \ra \geq (\mu -\varepsilon)\|\psi\|^2\,.
$$
Let $b>0$ be fixed and let $u,w : \R_+ \to [0,1]$ be functions from $\C^1(\R_+)$ satisfying $u^2+w^2=1$, $u(t)=1$ for $t\leq 1$, and $u(t) = 0$ for $t \geq (1+b)$. We choose the polar coordinates $(\rho, \varphi )$ centered at the asymptotes crossing and oriented in such a way that the asymptotes of $\LL$ correspond to $\varphi = \pm\beta$. Let $\psi \in \C_0^1 (\R^2)$ and consider the functions $U,W: \R^2 \to [0,1]$, where  $U=u (\rho a^{-1}), W=w(\rho a^{-1})$, and the constant $a$ will be chosen later. We set $\psi_0=\psi U$ and $\psi_1 = \psi W$. We have
\begin{eqnarray} \label{HVZeq2}
\lefteqn{L[\psi]:=\|\nabla \psi\|^2 +\|V^\frac12 \psi \|^2 -\alpha \int_\R |\psi(\LL(s))|^2\, \D s} \\ && \qquad
= L[\psi_0] +L[\psi_1] - \int_{\R^2} \{ |\nabla U|^2 + |\nabla W|^2 \} |\psi|^2 \, \D x \nonumber \\[.3em] && \qquad \geq L[\psi_0] +L[\psi_1] -
Ca^{-2}\|\psi\|^2 =L_1[\psi_0] +L_1[\psi_1] \nonumber
\end{eqnarray}
\Red{for some $C>0$}, where
$$
L_1[\phi] := L[\phi]- Ca^{-2}\|\phi\|^2 \,.
$$
Next we estimate the term $L_1[\psi_0]$. In the same ways as one checks the semi-boundedness of operator $H$ from below we get
$$
\frac12 \|\nabla \psi_0\|^2 +\|V^{1/2} \psi_0 \|^2 -\alpha \int_\R |\psi_0(\LL(s))|^2\, \D s \geq
C_0 \|\psi_0\|^2
$$
for some $C_0$. Notice that by construction $\rho \leq a (1+b)$ holds on the support of $\psi_0$ and recall that the Laplace operator on a disc has a purely discrete spectrum. Consequently, one can find a finite-dimensional subspace $M_1$ such that for all $\psi \perp M_1$ we have the inequality $\frac12 \|\nabla \psi_0\|^2 \geq \big(|C_0| +Ca^{-2}\big)\|\psi\|^2$. This yields the bound $L_1[\psi_0] \geq 0 \geq \mu\|\psi_0\|^2\,$. Let now $M := \{\phi |\: \phi U \in  M_1\}$. Obviously, $\psi _0 \perp M_1$ holds for all $\psi \perp M$ and the dimension of $M $ cannot then exceed that of $M_1$.

Our next aim is to show that $L_1[\psi_1] \geq (\mu -\eps)\|\psi_1 \|^2 $ holds for any fixed $\eps$ provided $a$ is chosen to large enough. To this end we will make a further partition of the unity in the configuration space, this time in the angular direction. We choose a $\gamma$ satisfying $0 < \gamma < \min \{ \frac{\pi}{2}-\beta,\  \beta\}$ and consider a function $u_1:\:[-\pi,\pi] \to [0;1]$ from $\C^1([-\pi, \pi])$ such that $u_1(\varphi)=1$ for $\varphi \in [\beta -\frac12 \gamma, \beta +\frac12 \gamma]$ and $u_1(\varphi)=0$ for $\varphi \notin [\beta - \gamma, \beta +\gamma]\,$. We put further $u_2 (\varphi) := u_1(-\varphi)$ and $u_0 (\varphi) := \sqrt{1-u_1^2(\varphi)- u_2^2(\varphi)}$. Then we have
\begin{eqnarray} \label{HVZeq3}
\lefteqn{\|\nabla \psi_1\|^2 \geq \|\nabla (\psi_1 u_1)\|^2 + \|\nabla (\psi_1 u_2)\|^2 +
\|\nabla (\psi_1 u_0)\|^2 - C \int |\psi_1|^2 \rho^{-2}\,\rho\, \D \rho\, \D \varphi} \\[.3em] && \qquad \geq
\|\nabla (\psi_1 u_1)\|^2 + \|\nabla (\psi_1 u_2)\|^2 +
\|\nabla (\psi_1 u_0)\|^2 - C a^{-2}\|\psi_1\|^2 \,. \phantom{AAAAAAAAAAA} \nonumber
\end{eqnarray}
Note that by construction the curve $\LL$ does not intersect the support of $\psi_1 u_0$ for all $a$ enough, which together with \eqref{HVZeq3} yields
$$
L_1[\psi_1] \geq \sum_{i=1}^2{L[\psi_1 u_i]} - \eps \|\psi_1\|^2 \,,
$$
where $\eps > 0$ can be chosen arbitrarily small for large $a$.

Let us now estimate the value of $L[\psi_1 u_1].$ Similary to the proof of the preceding lemma, we introduce the Cartesian coordinates with the origin at the asymptotes crossing and the $x$-axis coinciding with the asymptote $\varphi = \beta.$ Moreover, we assume that $a$ is chosen large enough, so that on the support of $\psi_1 u_1$ the curve $\LL$ is twice differentiable and given by the equation $y=l(x) .$ Then we have
\begin{eqnarray} \label{HVZeq4}
\lefteqn{L[\psi_1 u_1] \geq \int_{a\cos \gamma}^{\infty} \bigg\{\int_\R {\Big(\Big|\frac{\partial (\psi_1 u_1)}{\partial y}\Big|^2 +V(y)|\psi_1u_1|^2\Big)}\,\D y} \\[.3em] && \qquad \qquad -\alpha \sqrt{1+l'^2(x)}|\psi_1(x,l(x)) u_1(x,l(x))|^2 \bigg\}\, \D x\,. \nonumber
\end{eqnarray}
Recall that $l'(x) < \eps _1 $ holds for any fixed $\eps_1 > 0$ and a sufficiently large $a$ on the support of $\psi_1 u_1$. Applying Lemma~\ref{propChi} to the operator
$$
h = -\frac{\D^2}{\D y^2} -\alpha \sqrt{1+\eps_1}\delta (y) + V(y)
$$
and then integrating over $x$ we arrive at the estimates $L[\psi_1 u_1] \geq 0$ if $V_0^2 \geq (1 + \eps_1)\alpha^2 $, and $L[\psi_1 u_1] \geq -\frac14 \big(\alpha^2(1+\eps_1) -V_0 \big)^2\, \alpha^{-2}\,(1+\eps_1)^{-1}\, \|\psi_1u_1\|^2$ provided $V_0^2 < (1 + \eps_1)\alpha^2 .$ In both cases we thus get for all sufficiently small $\eps_1$ the bound
$$
L[\psi_1 u_1] \geq (\mu -\eps)\|\psi_1 u_1\|^2\,.
$$
In a similar way we check that $L[\psi_1 u_2] \geq (\mu -\eps)\|\psi_1 u_2\|^2$, which completes the proof.
\end{proof}

\section{An exterior positive potential}
\setcounter{equation}{0} \label{s:ext}

\noindent Our main question in this paper is about the conditions under which the operator \eqref{Ham} has or does not have eigenvalues below the threshold $\mu$. As we shall see the answer depends substantially on the orientation of the potential bias. By the assumption (a) above the curve $\LL$ separates two regions of which one is convex. We call the latter an interior region, denoted by $\mathcal{I}_\LL$, its complement is called exterior and denoted by $\mathcal{E}_\LL$.

\subsection{Example of a broken line}

\noindent Consider first the simplest case, when $\LL$ consists of two half-lines $AO$ and $OB$ with the a common origin at the point zero in the critical and supercritical case.
\begin{theorem} \label{main2}
Assume that $\angle AOB : = 2\beta < \pi $, and furthermore, that $V(x)= V_0 \ge \alpha^2$ holds outside $\angle AOB$ and $V(x)=0 $ inside the angle. Then $\sigma_\mathrm{disc}(H) = 0\,$.
\end{theorem}
\begin{proof}
To prove the claim it suffices to show that $\la H \psi, \psi \ra \ge 0$ holds on a core of the operator $H$, for instance, for any $\psi \in \C^2_0 (\R^2).$ Let $OC$ be the half-line with the endpoint $O$, orthogonal to $OB$, and directed into the halfplane opposite to the one containing $OA$. Similarly, let $OD$ be the half-line orthogonal to $AO$ emerging from $O$ and laying in the halfplane opposite to the one containing $OB$. By $\Omega _1$ and $\Omega _2$ we denote the regions inside the angles $\angle DOA$ and $\angle BOC$, respectively. Notice that they do not overlap due to the condition $2\beta < \pi$.  Obviously, for any $\psi \in \C^2_0 (\R^2)$ we have
\beq \label{est1}
\la H \psi, \psi \ra \ge \sum_{i=1}^2 \int_{\Omega _i} \big(|\nabla \psi |^2 + V_0 |\psi|^2 \big)(x)\, \D x - \alpha \int_\LL |\psi (\LL(s))|^2 \D s\,,
\eeq
because we have neglected a non-negative contribution from $\R^2\setminus(\Omega_1\cup\Omega_2)$. Moreover,
\begin{eqnarray*}
\lefteqn{
\int_{\Omega _1} \big(|\nabla \psi |^2 + V_0 |\psi|^2 \big)(x)\, \D x = \int_0^{\infty} \int_0^{\infty} \big( |\partial_{\tau}\psi|^2 +
|\partial_{s}\psi|^2 + V_0 |\psi|^2\big)(x)\, \D\tau \D s} \\ &&
\quad \ge \int_0^{\infty} \int_0^{\infty} \big(|\partial_{\tau}\psi|^2 +
 V_0 |\psi|^2\big)(x)\, \D\tau \D s \phantom{AAAAAAAAAAAAAAAAAAAAAAAAAA}
\end{eqnarray*}
where $\tau$ is the coordinate in $\Omega_1$ in the direction orthogonal to $AO$. This has to be compared with the appropriate part of the last term in \eqref{est1}. In view of Lemma~\ref{propFunc} we infer that
$$
\int_0^{\infty} \big(|\partial_{\tau}\psi|^2 + V_0 |\psi|^2\big)\, \D\tau \ge \sqrt{V_0}\,|\psi (\LL(s))|^2
$$
holds any $s$, which for $V_0 \ge \alpha^2$ implies
$$
\int_{\Omega _1} \big(|\nabla \psi |^2 + V_0 |\psi|^2 \big)(x)\, \D x -\alpha \int _{\LL\cap \mathrm{bd}(\Omega_1)} |\psi (\LL(s))|^2\, \D s \ge 0\,.
$$
A similar estimate can be made for $\Omega_2$ which yields the sought result.
\end{proof}

\subsection{Example continued: the subcritical case}

\noindent In the subcritical situation the spectral picture may change. We know from Ref.~\onlinecite{EI01} that a nontrivially curved and asymptotically straight leaky wire without a potential bias has always at least one bound state. Regarding our model as a perturbation of such a system one may expect that this will remain true if $V_0$ is small enough. It is indeed that case as the following claim shows.
\begin {proposition}
Let $\angle AOB = 2\beta < \pi$ be fixed. Then there exists a $V_c \in (0,\alpha^2)$ such that for all $ 0\le V_0 \le V_c$ the operator $H$ has at least one isolated eigenvalue below the threshold $\mu$ of its essential spectrum.
\end {proposition}
\begin{proof}
Notice that $\mu < 0$ holds for $V_0 < \alpha^2$, and consequently, it is not sufficient to construct a trial function $\psi$ such that $\la H \psi, \psi \ra < 0$. On the other hand, we have mentioned that the operator $H_0= H_{\alpha,\LL,0}$ without the potential bias has an eigenvalue $\lambda_0$ below $\mu_0 = \inf \sigma_\mathrm{ess} (H_0)$. Denoting by $\psi_0$ the corresponding eigenfunction, we have
$$
\la H \psi_0, \psi_0 \ra = \la H_0 \psi_0, \psi_0 \ra + \la V\psi_0, \psi_0 \ra \le (\lambda_0 + V_0)\| \psi_0 \|^2\,.
$$
The positivity of $V_0$ implies $\lambda_0 < \mu_0 < \mu $, which for $V_0$ small enough yields $\lambda_0 +V_0 < \mu$.
\end{proof}

In fact, one can choose for $\lambda_0$ the lowest eigenvalue of $H_0$ and obtain in this way an estimate of the critical value $V_c$. In particular, one can observe that the latter depends on the value of $\beta$. Our next goal is that is to show that for any subcritical value, $V_0 < \alpha ^2$, the system can have a discrete spectrum of any finite dimension provided $\beta$ is chosen small enough.

\begin {proposition} \label{smallangle}
Let $V_0 < \alpha^2$ be fixed. Then to any given $n\in\mathbb{N}$ there is a $\beta_n\in (0,\frac{\pi}{2})$ such that for all $ 0< \beta \le \beta_n$ we have $\sharp\sigma_\mathrm{disc}(H) \ge n$.
\end {proposition}
\begin{proof}
One obviously has $H \le H_0+V_0$, hence by the minimax principle\cite{RS} it is sufficient to show that $H_0+V_0$ has $n$ eigenvalues below $\mu$, the multiplicity being taken into account. However, this operator differs by a constant only from $H_0$ analyzed in Ref.~\onlinecite{EN03} and we can modify a variational argument used in that paper. Specifically, we choose the Cartesian coordinates with positive $x$-axis identified with the axis of the angle and $y$ perpendicular to it; the two half-lines they refer to the argument values $\pm\beta$. Our trial functions will be of the form
\beq \label{trial1}
\phi(x,y) = f(x)g(y)\,,
\eeq
where $f\in C^2(L,2L)$ is supposed to satisfy the condition $f(L)=f(2L)=0$ and $L$ is a parameter to be chosen later. Furthermore,
$$
g(y) =\left\{ \begin{array}{ll} 1 \quad & |y|\le 2d \\[.5em]
\eu^{\alpha(|y|-2d)} \quad & |y|\ge 2d \end{array} \right.
$$
with $d:= \tan\beta$. To prove the claim, we have to find an $n$-dimensional subspace spanned by functions of this type satisfying
$$
\|\nabla\phi\|^2 - \frac{2\alpha}{\cos\beta} \|f\|^2 +V_0\|\phi\|^2 < \mu \|\phi\|^2\,.
$$
Since $\|g\|^2 = 4d + \alpha^{-1}$ and $\|g'\|^2 = \alpha$, the norms are easily computed and the above condition is equivalent to
$$
- \left( \frac{ \alpha^2 - V_0}{2\alpha}\right)^2 > \frac{\alpha^2}{1+4d\alpha} \left( 1 - \frac{2}{\cos\beta} \right) + V_0 + \frac{\|f'\|^2}{\|f\|^2}\,,
$$
which can be further reformulated as
$$
\frac{4\alpha^4}{1+4d\alpha} \left( \frac54 - \frac{2}{\cos\beta} \right) + 4\alpha^2 \frac{\|f'\|^2}{\|f\|^2} + V_0^2 < 0\,.
$$
Choosing for $f$ functions from the span of the first $n$ eigenfunctions of the Dirichlet Laplacian on the interval $(L,2L)$, we achieve that the middle term at the left-hand side does not exceed $4\alpha^2\big( \frac{\pi n}{L} \big)^2$ and can be made arbitrarily small if $L$ is large enough. For a fixed choice of $L$ the left-hand side thus does not exceed
$$
-3\alpha^2 + 4\alpha^2 \left(\frac{\pi n}{L} \right)^2 + V_0^2
$$
in the limit $\beta\to 0$, and since $V_0<\alpha^2$, it is negative for all $\beta$ small enough.
\end{proof}

\subsection{More general curves}

\noindent Our next goal is to generalize the main result of this section, Theorem~\ref{main2}, to a wider class of curves specified by the assumptions (a)--(d) of Sec.~\ref{s:prelim}.

\begin{theorem} \label{main3}
For the described curve class, let $V(x)= 0$ in $\mathcal{I}_\LL$ and $V(x) = V_0 \ge \alpha^2$ otherwise. Then $\sigma (H) = [0,\infty)\,$.
\end{theorem}
\begin{proof}
In view of Theorem~\ref{main1} it suffices to check that for any $\psi \in \C_0^2 (\R^2)$ we have $\la H\psi , \psi \ra \ge 0$. Neglecting the non-negative integral over $\mathcal{I}_\LL$, we can write the inequality
\beq
\label{eqTh3Eq1}
\la H\psi, \psi \ra \ge I[\psi]:= \int_{\mathcal{E}_\LL} \left(|\nabla \psi |^2 + V_0 |\psi|^2 \right)(x)\, \D x - \alpha \int_\LL {|\psi(\LL(s))|^2}\, \D s\,.
\eeq
Let $\LL_j := (s_j,s_{j+1}) \rightarrow \R^2,\: j=1,2,\dots,n-1$ with $s_1 = -\infty$ and $s_n = \infty$ be the $\C^2$ segments of $\LL$. For $j=2,\dots,n-1$, let $\Gamma_j^1$ and $\Gamma_j^2 $ be the half-lines in $\mathcal{E}_\LL$ emerging from the point $\LL(s_j) $ and orthogonal at this point to the tangents to $\LL_{j-1}$ and $\LL_j$ respectively (see Fig.~1). Denote by $\Omega_1$ the regions with the boundaries $\LL_1$ and $\Gamma^2_2$, by $\Omega_j, \: j=2,\dots,n-1$, the regions with the boundaries $\Gamma_j^2, \LL_j, \Gamma_{j+1}^1$, and by $\Omega_{n-1}$ the region bounded by $\Gamma_{n-1}^2$ and $\LL_{n-1}$.

\begin{figure}
\label{illustration}
\definecolor{qqqqff}{rgb}{0.,0.,1.}
\begin{tikzpicture}[line cap=round,line join=round,x=.7cm,y=.7cm]
\draw [shift={(6.40598984771574,-3.4676751269035613)}] plot[domain=1.8136705445139667:2.3097943466025104,variable=\t]({1.*7.343192279718322*cos(\t r)+0.*7.343192279718322*sin(\t r)},{0.*7.343192279718322*cos(\t r)+1.*7.343192279718322*sin(\t r)});
\draw [shift={(12.273844282238448,-4.029343065693432)}] plot[domain=2.635791642892342:3.0094021834296742,variable=\t]({1.*12.361693189813085*cos(\t r)+0.*12.361693189813085*sin(\t r)},{0.*12.361693189813085*cos(\t r)+1.*12.361693189813085*sin(\t r)});
\draw (1.46,1.96)-- (-0.4708436835698764,4.078886725999075);
\draw (1.46,1.96)-- (-0.9317305985500925,3.284681094119432);
\draw [shift={(3.814170854271357,-4.972562814070351)}] plot[domain=0.9036391554834776:1.4754221412880848,variable=\t]({1.*8.671974084186662*cos(\t r)+0.*8.671974084186662*sin(\t r)},{0.*8.671974084186662*cos(\t r)+1.*8.671974084186662*sin(\t r)});
\draw (4.64,3.66)-- (4.173380056039653,5.543315112250095);
\draw (4.64,3.66)-- (4.829610719656421,5.642040057705814);
\draw [shift={(-5.113421052631559,-4.717631578947365)}] plot[domain=0.11968280059086746:0.43013688237047787,variable=\t]({1.*15.725915468201485*cos(\t r)+0.*15.725915468201485*sin(\t r)},{0.*15.725915468201485*cos(\t r)+1.*15.725915468201485*sin(\t r)});
\draw (9.18,1.84)-- (11.060629730014458,4.227684702173255);
\draw (9.18,1.84)-- (12.328922116487139,3.2846836089614855);
\draw (-0.42,2.68) node[anchor=north west] {$\Gamma^1_2$};
\draw (0.02,-2.4)-- (-0.3550512469390227,-5.220658015256111);
\draw (10.5,-2.84)-- (10.782898164766083,-5.192436021546814);
\draw (0.42,4.1) node[anchor=north west] {$\Gamma^2_2$};
\draw (-0.24,1.66) node[anchor=north west] {$\Omega_1$};
\draw (1.96,4.14) node[anchor=north west] {$\Omega_2$};
\draw (7.28,4.56) node[anchor=north west] {$\Omega_3$};
\draw (10.84,1.02) node[anchor=north west] {$\Omega_4$};
\draw (0.96,0.32) node[anchor=north west] {$\mathcal{L}_1$};
\draw (3.08,2.86) node[anchor=north west] {$\mathcal{L}_2$};
\draw (6.5,2.78) node[anchor=north west] {$\mathcal{L}_3$};
\draw (9.16,-0.14) node[anchor=north west] {$\mathcal{L}_4$};
\begin{scriptsize}
\draw [fill=qqqqff] (1.46,1.96) circle (1.5pt);
\draw [fill=qqqqff] (4.64,3.66) circle (1.5pt);
\draw [fill=qqqqff] (9.18,1.84) circle (1.5pt);
\end{scriptsize}
\end{tikzpicture}
\caption{To the proof of Theorem~\ref{main3}}
\end{figure}
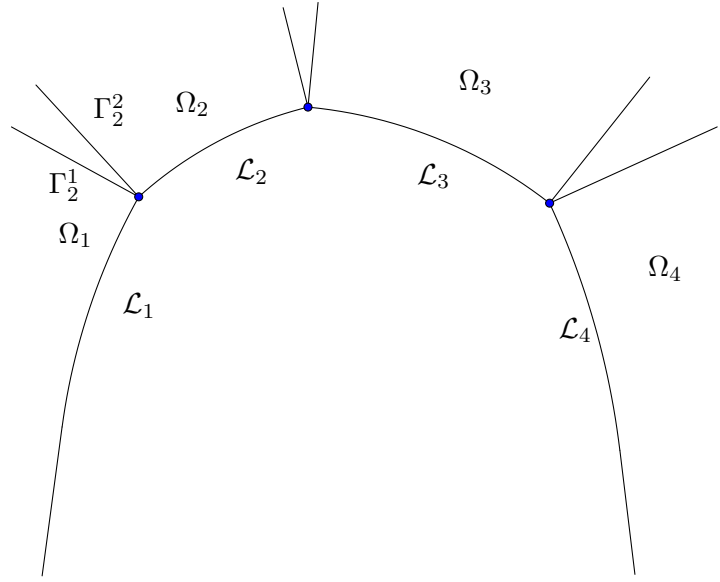

In view of the assumed convexity of $\mathcal{I}_\LL$ the regions $\Omega_j$ do not overlap mutually. This yields
\beq
\label{eqTh3Eq2}
 I[\psi] \ge \sum_{j=1}^{n-1} \int_{\Omega_j} \left(|\nabla \psi |^2 + V_0 |\psi|^2 \right)(x)\, \D x - \alpha \sum_{j=1}^{n-1} \int_{s_j}^{s_{j+1}} {|\psi(\LL(s))|^2}\, \D s\,.
\eeq
We introduce in $\Omega_j$ locally orthogonal coordinates $(s,\tau)$, where $\tau$ is measured in the normal direction to the curve at the point $\LL(s)$. Writing with an abuse of notation $\psi(\LL(s))$ as $\psi(s,0)$, we can we estimate the $j$-th term in (\ref{eqTh3Eq2}) in the following way,
\begin{eqnarray}
\label{eqTh3Eq3}
\lefteqn{\int_{\Omega_j} \left(|\nabla \psi |^2 + V_0 |\psi|^2 \right)(x)\, \D x - \alpha \int_{s_j}^{s_{j+1}} {|\psi(\LL(s))|^2}\, \D s} \\ &&
\quad \ge \int_{s_j}^{s_{j+1}} \left\{
{\int_0^{\infty}{\Big[ \Big(\frac{\partial \psi}{\partial \tau}\Big)^2 +V_0 |\psi|^2 \Big]\Big (1 + \frac{\tau}{R(s)}\Big) \D\tau} }
-\alpha |\psi(s,0)|^2 \right\}\,\D s\,, \nonumber
\end{eqnarray}
by neglecting the non-negative term containing $\big(\frac{\partial \psi}{\partial s}\big)^2$; here $R(s) $ is the curvature radius of $\LL$ at the point $s$. According to Lemma~\ref{propFunc} the expression in the curly bracket on the right-hand side of \ref{eqTh3Eq3} is non-negative for any $s\in\R$, even if the curve is locally straight and $R(s)=\infty$. This implies the non-negativity of $I[\psi]$.
\end{proof}

\section{An interior positive potential}
\setcounter{equation}{0} \label{s:int}

\noindent The results of the previous section say, in particular, that in the critical case a curvature does not give rise to bound states provided the potential bias is located in the exterior region. Let us inspect now the opposite situation.

\subsection{The broken line example revisited}

\noindent Let first $\LL$ be a broken line such as we have considered in Theorem~\ref{main2} now assuming $V(x) = 0$ in $\mathcal{E}_\LL$. We will start with the critical case when $V_0 = \alpha^2$ in $\mathcal{I}_\LL$.

\begin{theorem} \label{main4}
$\,\sigma_\mathrm{disc} (H) \neq \emptyset\,$ holds in the described situation.
\end{theorem}
\begin{proof}
According to Theorem~\ref{main1} we have $\sigma_\mathrm{ess} (H) = [0, \infty)$. To prove the claim, it thus suffices to find a $\psi \in \C_0^2(\R^2)$ such that $\la H\psi, \psi \ra < 0$ holds. We shall work in a setting similar to that of the proof of Theorem~\ref{main3}, however, using now polar coordinates $(\rho, \varphi)$ instead of the Cartesian ones. The two half-lines separating the regions $\mathcal{I}_\LL$ and $\mathcal{E}_\LL$ are then $\big\{ (\rho, \varphi)| \ 0 \le \rho < \infty,\: \varphi = \pm \beta\big\}$. We define first the trial function $\psi(x)$ in the exterior region, setting for $\big\{ -\pi \le \varphi \le - \beta \big\} \bigcup \{ \beta \le \varphi \le \pi \big\}$ its values to be $\psi (\rho, \varphi) = 1$ if $0 \le \rho < a$, $\,\psi (\rho, \varphi) = \left( \ln\frac ba \right)^{-1} \ln \left( \frac{\rho}b\right)$ for $a \le \rho < b$, and finally $\psi (\rho, \varphi) = 0$ for $\rho \ge b$, where the constants $b > a > 0$ will be chosen later. On the other hand, for $ - \beta < \varphi <  \beta$ we define the function $\psi$ by
$$
\psi(x) = \eu^{-\alpha \cdot \mathrm{dist}(x, L)} \psi(\LL(s(x)))\,,
$$
where $L(s(x))$ is a point on $L$ nearest to $x$. The latter may be not unique, of course, but it is not difficult to see that $\psi(x)$ is correctly defined.
			
Since $V(x) = 0$ holds in $\mathcal{E}_\LL$, the `volume' part of the form needs in this region to estimate the kinetic energy only. It can be computed explicitly,
$$
\int_{\mathcal{E}_\LL} |\nabla \psi|^2(x)\,\D x = (2\pi - 2\beta) \left[ \ln \left(\frac ba \right)\right]^{-2}\cdot
			\int_a^b \frac 1{\rho^2}\, \rho \  \D\rho = (2\pi - 2\beta)\left[ \ln \left(\frac ba \right)\right]^{-1},
$$
hence choosing $\frac ba$ sufficiently large, one is able to make this term arbitrarily small. To estimate the contribution to $\la H\psi, \psi \ra$ from $\mathcal{I}_\LL$ we cut this region into three parts. Let $\Gamma_a^\pm$ and $\Gamma_b^\pm$ be the perpendiculars to $\LL$ in $\mathcal{I}_\LL$ at the points $(a,\pm\beta)$ and $(b,\pm\beta)$,  respectively. By $\Omega_1$ we denote    the part of $\mathcal{I}_\LL$ bounded by the (appropriate part of) $\LL$ and $\Gamma_a^\pm$. Let further $\Omega_2$ be the part of $\mathcal{I}_\LL$ with the boundary consisting of the (appropriate part of) $\LL$ together with $\Gamma_a^\pm$ and $\Gamma_b^\pm$, the rest of $\mathcal{I}_\LL$ we will be called $\Omega_3$. Notice that, by construction of $\psi$ in the exterior region, we have $\psi(x) = 0$ in $\Omega_3$, hence we need to estimate the quadratic form of $H$ in $\Omega_1$ and $\Omega_2$ only.
			
In $\Omega_1$ the gradient $\nabla \psi$ is parallel to $\Gamma_a^\pm$ for $\pm\varphi>0$. Moreover, $\psi(\LL(s(x))) = 1$ holds for $x \in \Omega_1$. If we thus introduce the coordinates $(s,\tau)$, where $\tau$ is the distance to $L$, we get $\psi(x) \equiv \psi(s,\tau) = \eu^{-\alpha \tau}$. Using this observation, we immediately arrive at
\begin{eqnarray*}
\lefteqn{\int_{\Omega_1} \left(|\nabla \psi |^2 + V_0 |\psi|^2 \right)(x)\, \D x = 2\int_0^{a}{\int_0^{s\tan\beta} {2\alpha^2 \eu^{-2\alpha \tau}\, \D\tau}\D s}} \\ && \qquad\qquad < 2\alpha a = \alpha \int_{\LL \cap \mathrm{bd}(\Omega_1)} {|\psi (\LL(s))|^2 \D s}\,. \phantom{AAAAAAAAAAAAAAA}
\end{eqnarray*}
The last inequality shows that $\la H \psi, \psi \ra _{\Omega_1} = \gamma < 0 $ with some $\gamma$ independent of $b\,$; here and below
$\la H \psi, \psi \ra _{\Omega_i}$  are the parts of the quadratic form of $H$ corresponding to $\Omega_i,\:i=1,2.$ To complete the proof of Theorem~\ref{main4} it suffices to demonstrate that $\la H \psi, \psi \ra _{\Omega_2} \to 0$ holds as $b \to \infty.$ Due to the mirror symmetry with respect to the angle exis it is enough to compute the corresponding integrals over $\Omega_2 \cap \{ \varphi > 0\}.$ In the variables $(s,\tau)$ the function $\psi$ can be written as
$$
\psi (s,\tau) = \eu^{-\alpha \tau}\left(\ln {\frac ba} \right)^{-1} \ln {\frac sb} \quad\text{with}\quad a \le s \le b\,, \ 0 \le \tau \le s \tan \beta\,.
$$
Consequently, we have
\begin{eqnarray*}
\lefteqn{\la H \psi, \psi \ra _{\Omega_2} = 2\left(\ln {\frac ba}\right)^{-2}\int_a^{b}{2\alpha^2 \left[ \ln^2 {\frac sb} + s^{-2} \right]
	\int_0^{s\tan\beta}{\eu^{-2\alpha \tau}\, \D\tau} \D s}} \\ &&  \qquad \le 2\alpha \left(\ln {\frac ba}\right)^{-2}\int_a^b {\ln ^2 {\frac sb}\, \D s}  \le  \frac12 \alpha \left(\ln {\frac ba}\right)^{-2}(a^{-1}
-b^{-1}) \rightarrow 0
\end{eqnarray*}
as $b \to \infty $ and $a$ is fixed, which is what we have set out to prove.
\end{proof}

\subsection{Example continued: the supercritical case}

\noindent Next we are going to prove that $\alpha^2$ is indeed the critical potential value from the viewpoint of curvature-induced bound states. At the same time we shall obtain a certain counterpart to Proposition~\ref{smallangle}.

\begin{theorem} \label{main5}
In the same situation as in Theorem~\ref{main4}, let $V_0 > \alpha^2.$ Then $\sigma_\mathrm{disc}(H) = \emptyset$ holds for any $\beta > \frac{\pi}{2} V_0^{-\frac 12}\alpha$.
\end{theorem}
\begin{proof}
Suppose that our operator $H$ with a fixed coupling parameter $\alpha$ and angle $2\beta $ satisfying $\pi >2\beta > \pi V_0^{-\frac 12}\alpha$  has an eigenvalue. Recall that due to the assumption $V_0 > \alpha^2$ we have $\sigma_\mathrm{ess}(H) = [0,\infty)$, hence there is a $\psi_0\in H^1(\R^2) \cap L^2(\R^2)$ such that
$$
I[\psi_0] := \| \nabla \psi_0\|^2_{\mathcal{I}_\LL} +V_0 \| \psi_0\|^2_{\mathcal{I}_\LL} - \alpha \int_\LL {|\psi_0(\LL(s))|^2\,\D s} < 0\,,
$$
where $\|.\|_{\mathcal{I}_\LL}$ is the $L^2$-norm over the region $\mathcal{I}_\LL$. Let $\Omega_0$ be the right half-plane,
$$
\Omega_0 := \left\{ (\rho, \varphi) |\: \rho \in [0, \infty)\,, \ \varphi \in \big[-\frac{\pi}{2}, \frac{\pi}{2}\big]\right\}\,,
$$
and let $\LL_0$ be its boundary. We define function $\psi:\Omega_0 \to \R^2$ by an angular rescaling of $\psi_0$, namely $\psi(\rho, \varphi):= \psi_0 (\rho, 2\beta \pi ^{-1}\varphi)$. Then we have
$$
\|\nabla \psi\|^2_{\Omega_0} \le \frac12\pi\beta ^{-1}|\nabla \psi_0\|^2_{\mathcal{I}_\LL}\,, \ \ \ \   V_0\| \psi\|^2_{\Omega_0} = \frac12 V_0\pi \beta^{-1}\|\psi_0\|^2_{\mathcal{I}_\LL}\,,
$$
$$
 \int_{\LL_0}{|\psi(\LL_0(s))|^2\,\D s} = \int_{\LL}{|\psi_0(\LL(s))|^2\,\D s}\,.
$$
Using these relations and choosing $\alpha_1 = \frac12\pi \beta ^{-1}\alpha$ we get
$$
\| \nabla \psi\|^2_{\Omega_0} +V_0 \| \psi\|^2_{\Omega_0} - \alpha_1 \int_{L_0} {|\psi(L(s))|^2 ds} \le \frac12\pi \beta^{-1} I[\psi_0]<0\,,
$$
however, the last inequality contradicts Lemma~2, because $\alpha_1 < \sqrt{V_0}$ by assumption.
\end{proof}

\subsection{More general curves}

\noindent Finally, let us generalize the claim of Theorem~\ref{main4} to a wider class of curves similarly as we did it at the end of previous section; we shall again assume that $\LL$ satisfies the assumptions (a)--(d) of Sec.~\ref{s:prelim}. Here we add a restriction which would allow us to simplify further geometrical arguments: we shall suppose that outside a compact, i.e. for all sufficiently large $|s|$ the curve $\LL$ coincides with its asymptotes.

\begin{theorem} \label{main6}
In the described situation, let $V(x) = V_0 =\alpha^2$ hold for $x \in \mathcal{I}_\LL$ and $V(x)=0$ otherwise. Then $\sigma_\mathrm{disc}(H) \neq \emptyset$.
\end{theorem}
\begin{proof}
To prove the claim it suffices to find a $\psi \in \C_0^1(\R^2)$ such that $\la H \psi, \psi \ra < 0$.  Let $O$ be the point, where the asymptotes intersect each other. We will choose this point as the origin of the polar system of coordinates. In analogy with the previous considerations assume that the asymptotes in this coordinate system are
$$
	T_\pm := \left\{ (\rho, \varphi) |\: \rho \in [0,\infty)\,,\: \varphi =\pm\frac12 \beta \right\}\,.
$$
Let $b>a>0$ be chosen so large that it holds $(a,\pm\frac{\beta}{2}) \in T_\pm \cap \LL$. We define function $\psi$ in a way similar to the proof of Theorem~\ref{main4}, setting $\psi (\rho, \varphi) = 1$ for $(\rho, \varphi) $ in $\mathcal{E}_\LL$ with $\rho < a$. If the point $(\rho, \varphi) $ belongs to $\mathcal{E}_\LL$, but $a \le \rho \le b$ we set $\psi (\rho, \varphi) = \left(\ln \frac ab \right)^{-1}\ln \frac {\rho}{b}$, while in the remaining part of $\mathcal{E}_\LL$ the function will vanish, $\psi (\rho, \varphi) = 0$. Following the proof of Theorem~\ref{main4} we define $\psi(x)$ in $\mathcal{I}_\LL$ as $\psi(x) = \eu^{-\alpha \cdot \mathrm{dist}(x, \LL)} \psi(\LL(s(x)))$ and check that 
$$
	\|\nabla \psi\|^2_{\mathcal{E}_\LL} + V_0\| \psi\|^2_{\mathcal{E}_\LL} \le 2\pi \left(\ln \frac ba \right)^{-1} \rightarrow 0
$$
holds as $ ba^{-1} \to \infty$. Following the same line of reasoning we define the perpendiculars $\Gamma_a^\pm,\, \Gamma_b^\pm$ and the regions $\Omega_i,\: \ i=1,2,3$, in the same way as before, and set
$$
		I[\psi]_i :=\la H\psi, \psi \ra _{\Omega _i} = \|\nabla \psi\|^2_{\Omega_i} + V_0\| \psi\|^2_{\Omega_i} - \alpha \int_{\LL\cap \mathrm{bd}(\Omega_i)}{|\psi(\LL(s))|^2\,\D s}\,.
$$
It is obvious that $I[\psi]_3 =0$, the estimate of $I[\psi]_2$ is not different from the one given in the proof of Theorem~\ref{main4}. To complete the proof it suffices therefore to show that $I[\psi]_1 \le \gamma < 0 $ for some constant $\gamma$ independent of $b$.
	
Let $\LL_1 = \LL \cap \Omega_1$ and let $\LL_1^i : (s_i, s_{i+1}) \to \R^2\,,\: 0=s_1 < s_2 <\dots<s_n$, be the $\C^2$ segments of $\LL_1$ such that $\bigcup_{i=1}^n \LL_1^{i} = \LL_1$. For $x \in \Omega_1 \setminus \LL_1 $ let $l(x)$ be a point in $\LL$ with the minimum distance to $x$. Again, such a point may be not unique, but this fact paly no role in the estimates below. We are going to use the following simple geometrical properties of $\Omega_1 $ related mainly to the convexity of the set and to the orthogonality of $\Gamma_a^\pm$ to $\LL$.
	\begin{enumerate}[(i)]
\item $\,l(x) \in \LL_1$ for any $x \in \Omega_1 \setminus \LL_1$.
\item Let $ \Gamma_x$ be the straight line segment connecting $x$ with $l(x)$. Then $\Gamma_x$ is orthogonal to the tangent to $\LL_1$ at the point $l(x)$.
\item Let $y \in \Gamma_x$ be a point on $\Gamma_x$, which lies between $x$ and $l(x)$. Then $l(y)=l(x)$.
\item Let $R(l(x))$ be the curvature radius of $\LL_1$ at the point $l(x)$. Then $\mathrm{dist}(x, l(x)) \le R(l(x))$.
	\end{enumerate}
The last property can be easily proved by \emph{reductio ad absurdum}. Indeed, should the inequality $\mathrm{dist}(x, l(x)) > R(l(x))$ hold, there would be a point $l_1(x)$ in a small neighborhood of $l(x)$ such that $\mathrm{dist}(x, l(x)) > \mathrm{dist}(x, l_1(x))$, but this contradicts the definition of $l(x)$.

Notice that property (i) yields $\psi (x) = 1$ for any $x \in \LL_1$. Using this fact we can simplify the expression of $\psi (x)$  for all $x \in \Omega_1$, writing it as $\psi (x)= \eu^{-\alpha \cdot \mathrm{dist} (x, l(x))}$. Obviously, the gradient $\nabla \psi (x)$ is parallel to $\Gamma_x$ which implies
$$
|\nabla \psi (x)| = \alpha |\psi (x)| = \alpha \eu^{-\alpha \cdot \mathrm{dist} (x, l(x))}\,.
$$
For an arbitrary $\LL(s) \in \LL_1$ we set
$$
M(s):= \big\{ x|\: x \in \Omega_1 \setminus  \LL,\,  l(x)=\LL(s)\, \big\}, \quad\text{ad}\quad T(s):= \max_{x \in M(s)} \mathrm{dist} (x, \LL(s))\,.
$$
In a way similar to the proof of Theorem~\ref{main4}, we estimate $I[\psi]_1$ by introducing the coordinates $(s,\tau)$, where $\tau$ is measured in the direction of the normal vector to $\LL$ at the point $\LL(s)$. This gives
\beq
\label{eqLast}
  I[\psi]_1 = \sum_{i=1}^{n}{\int_{s_i}^{s_{i+1}}{\Big[ 2\alpha^2 \int_0^{T(s)}{\big( 1- {\tau}{R(s)}^{-1} \big)\, \eu^{-2\alpha \tau} d\tau} -\alpha \Big]\,\D s}\,.}
\eeq
	The integral in (\ref {eqLast}) is strictly negative and independent of $b$, because the integration over $\tau $  goes up to $T(s)$ and $T(s)$ is a finite number.
\end{proof}

\section{Finiteness of the discrete spectrum}
\setcounter{equation}{0} \label{s:finite}

\noindent In the closing section return to the question about cardinality of the discrete spectrum. We have shown in Proposition~\ref{smallangle} that the number of eigenvalues can be large in suitable geometries, now we are going to show that it will be always nevertheless finite, at least if we add a rather mild assumption on how the curve $\LL$ approaches its asymptotes. Specifically, let $d(s)$ be the distance of the point $\LL(s)$ to the nearest asymptote; we will suppose that $d'(s)=o(|s|^{-1})$ as $|s| \to \infty .$ This hypothesis in indeed not very restrictive, recall that  in view of the conditions (a)--(d) of Sec.~\ref{propL} the function $d'(\cdot)$ is monotone for large $|s|$ and the integral  $\int_{\R} |d'(s)|\,\D s $ converges.

\begin{theorem} \label{fin}
In the described situation we have  $\sharp\sigma_\mathrm{disc} (H) < \infty\,$.
\end{theorem}
\begin{proof}
A different argument is needed in two different cases. $\mu<0$ and $\mu=0$. Let us start with the subcritical situation, $\mu < 0$, when we will refine the technique developed in the proof of Lemma~\ref{HVZ2}. Specifically, we are going to construct a finite dimensional subspace $M$, such that $\la H\psi, \psi \ra \geq \mu \|\psi\|^2$ holds for any $\psi \in \C^2_0(\R^2)$ such that $\psi \perp M$. The notations will be the same as in the proof of Lemma~\ref{HVZ2}.

Inspired by Ref.~\onlinecite{Zh74} we consider functions $u$ and $w$ defined in the proof of Lemma~\ref{HVZ2} requiring them to satisfy two additional conditions, namely $w(t)<1$ for $t<1+b$ and $\lim_{t \to (1+b)-} w'^2(t)(1-w^2(t))^{-1} = 0$.  The last requirement implies, in particular, that we can find a $b_1 \in (0,b)$ such that $[u'^2(t) + w'^2(t)]w^{-2}(t) \leq \eps$ holds for all
$t \in [1+b_1, 1+b]$. On the other hand, on $[1, 1+b_1]$ we have $u(t) > 0$ and one can find a constant $C_0$ depending on $b_1$ and $\eps$ such that $u^2(t)+w'^2(t) \leq C_0 u^2(t)$ holds on this interval. Applying these observations to functions $U=u(\rho a^{-1})$ and $W= w(\rho a^{-1})$ gives (compare with inequality \eqref{HVZeq2})
\begin{eqnarray}
\label{eqfin0}
\lefteqn{\int \big\{ | U |^2 + |W|^2 \big\} |\psi|^2\, \D \omega = C_0 \int \int_{\rho \in [a,a(1+b) ]} a^{-2}|\psi U|^2\, \rho\, \D \rho\, \D \varphi} \\[.3em] && \qquad\qquad +  \eps \int \int_{\rho \in [a,a(1+b) ]} a^{-2}|\psi W|^2\, \rho\, \D \rho\, \D \varphi\,, \nonumber \phantom{AAAAAAAAAAAAAA}
\end{eqnarray}
This yields
\beq
\label{eqfin1}
 L[\psi] \geq L_1[\psi_0] + L_2[\psi_1] \,,
\eeq
where $L[\psi]$ and $L_1[\psi_0]$ were defined in the proof of Lemma~\ref{HVZ2}, $\psi_1 := \psi W$, and
$$
L_2[\psi_1] := L[\psi_1] -\eps (1+b)^2\|\psi_1 \rho^{-1}\|^2 \,.
$$
Repeating the arguments used in the proof of Lemma~\ref{HVZ2} we infer that there is a finite dimensional subspace $M$ such that
$$
L_1[\psi_0] \geq \mu \|\psi_0\|^2
$$
holds for all $\psi \perp M$. To prove the claim it suffices now to show that $L_2[\psi_1] \geq \mu \|\psi_1\|^2$ holds for a sufficiently large $a>0$ and all $\eps$ small enough.

As in the proof of Lemma~\ref{HVZ2} we introduce next the angular partition determined by the functions $u_1,\, u_2$, and $u_0$ assuming that, in addition to the properties described before, $u_1$ satisfies the condition
\beq
\label{eqfin2}
 \lim_{\varphi \to (\beta +\frac{\gamma}{2})+}\: \frac{u_1'^2(\varphi)}{1-u_1^2(\varphi)} =
\lim_{\varphi \to (\beta -\frac{\gamma}{2})-}\> \frac{u_1'^2(\varphi)}{1-u_1^2(\varphi)} = 0\,.
\eeq
With this additional requirement analogous to inequality \eqref{eqfin0} we can estimate the so-called localization error by $\eps \rho^{-2}$ on the supports of $u_1$ and $u_2$, and by $C(\eps)\rho^{-2}$ on the support of $u_0.$ This yields
\beq
\label{eqfin3}
 L_2[\psi_1] \geq \sum_{i=1}^2 L_3[\psi_1u_i] + L_4[\psi_1u_0] \,,
\eeq
where
\beq
\label{eqfin4}
 L_3[\phi] := \|\nabla \phi\|^2 +\|V^{1/2}\phi\|^2-\alpha \int |\phi(\LL (s))|\,\D s - 2\eps
(1+b)^2\, \big\|\phi \rho^{-1}\chi_{a \leq \rho \leq a(1+b)}\big\|^2
\eeq
and
\beq
\label{eqfin5}
 L_4[\phi] := \|\nabla \phi\|^2 +\|V^{1/2}\phi\|^2 - C\|\phi \rho^{-1}\|^2\,,
\eeq
provided that $a$ is chosen sufficiently large to ensure that the curve does not intersect the support of $\psi u_1 .$ Notice that the constant $C$ in \eqref{eqfin5} depends on $\eps$, but is independent of $a$, and that the function $\psi_1 u_0$ vanishes for $\rho < a$. Consequently, for all large enough $a$ we have
\beq
\label{eqfin6}
 L_3[\psi_1 u_0] \geq \mu \|\psi_1 u_0\|^2\,.
\eeq
To complete the proof in the subcritical case it suffices now to show that
$$
 L_4[\psi_1 u_i] \geq \mu \|\psi_1 u_i\|^2\,,\ \ \ i=1,2 \,.
$$
We will prove this inequality for $i=1$, the proof for $i=2$ is similar.

Denote the product $\psi_1 u_1$ by $\psi_3 .$ This function is supported in the truncated wedge $\{(\rho, \varphi)|\: \rho \geq a\,,\;  \beta -\gamma \leq \varphi \leq \beta +\gamma \}.$ In analogy with the proof of Lemma~\ref{HVZ2} we introduce the Cartesian coordinate system with the origin at the asymptote crossing, $\rho = 0$, and the $x$-axis coinciding with the asymptote $\varphi = \beta .$ Then we get
\begin{eqnarray}
\label{eqfin7}
\lefteqn{L_3[\psi_3] \geq \int_{a\cos \gamma}^{\infty}\bigg\{ \int \Big\{\Big|\frac{\partial \psi_3}{\partial y}\Big|^2 +V(y)|\psi_3|^2 \Big\}\, \D y - \alpha \sqrt{1+l'^2(x)}\,|\psi_3(x,l(x))|^2 \bigg\}\,\D x } \\[.3em] &&
\qquad\qquad +
\int\bigg\{ \int_{a\cos \gamma}^{\infty}\Big|\frac{\partial \psi_3}{\partial x}\Big|^2 - \frac{2\eps (1+b)^2}{|x|^2}|\psi_3|^2   \bigg\}\, \D x \D y\,. \nonumber \phantom{AAAAAAAAAAAAAAA}
\end{eqnarray}
The first term on the right-hand side of \eqref{eqfin7} can be estimated using Lemma~\ref{propChi}(iii),
\begin{eqnarray}
\label{eqfin7.1}
\lefteqn{\int_{a\cos \gamma}^{\infty}\bigg\{ \int \Big\{\Big|\frac{\partial \psi_3}{\partial y}\Big|^2 +V(y)|\psi_3|^2 \Big\}\, \D y - \alpha \sqrt{1+l'^2(x)}\,|\psi_3(x,l(x))|^2 \bigg\}\,\D x } \\[.3em] &&
\geq
 -\int_{a\cos \gamma}^{\infty}\bigg\{\frac14\big[\alpha^2(1+l'^2(x))-V_0\big]^2\alpha^{-2}\, [1+l'^2(x)]^{-1}\int |\psi_3(x,y)|^2\, \D y\bigg\}\, \D x \,.\nonumber \phantom{AAAAA}
\end{eqnarray}
Recall that $l'(x) = o (|x|^{-1})$ holds by assumption as $|x| \to \infty .$ This implies that
$$
-\frac14 \Big[\alpha^2(1+l'^2(x))-V_0\Big]^2\,\alpha^{-2}\, [1+l'^2(x)]^{-1} = -\frac14(\alpha^2 -V_0)^2\alpha^{-2}+
o(|x|^{-2}) = \mu +o(|x|^{-2})\,.
$$
Substituting form here into the above estimate we arrive for a large enough $a$ at
$$
L_3[\psi_3] \geq \mu \|\psi_3\|^2 + \int\bigg\{ \int_{a\cos \gamma}^{\infty} \Big[
\Big|\frac{\partial \psi_3}{\partial x}\Big|^2 - \frac{3\eps (1+b)}{|x|^2} |\psi_3|^2\Big] \D x\bigg\}\, \D y \,.
$$
Due to the Hardy inequality the inner integral is positive, which yields
$$
L_3[\psi_3] \geq \mu\|\psi_3\|^2
$$
completing thus the proof in the subcritical case, $\mu < 0 .$

Let us turn now to the opposite situation with $V_0 \geq \alpha^2$ where, as we know, $\mu=0$. In view of Theorem~\ref{main3} the discrete spectrum is empty if $V(x) = 0$ in $\mathcal{I}_\LL$, hence we only need to consider the case when the potential vanishes in $\mathcal{E}_\LL$. Similarly to the case $\mu < 0$ one can see that the statement of the theorem is true if
$$
L_2[\psi] \geq \mu \|\psi_1\|^2 = 0
$$
holds for all sufficiently large $a>0$ and all $\psi_1 \in \C_0^2(\R^2)$ satisfying $\psi_1 = 0$ for $\rho \leq a .$

Let functions $u^*,\,w^*:\:[-\pi, \pi] \to [0, 1 ]$ obey the following conditions:
 \begin{enumerate}[(a)]
\item
$ (u^*)^2 +(w^*)^2 = 1 \,,$
\item
$ u^*,\,w^* \in \C^1([-\pi, \pi])\,,$
\item
$u^*(\varphi) = 0 $ for $\varphi \in [-\pi, -\gamma] \cup [\gamma, \pi] $ and $u^*(\varphi) = 1$ for
$\varphi \in [-\frac{\gamma}{2}, \frac{\gamma}{2}] \,,$
\item
$\lim _{|\varphi| \to \gamma-} \big((w^*)'\big)^2\, \Red{\big(1-(w^*)^2\big)^{-1}} = 0\,$.
\end{enumerate}
Following the same line of arguments as in the case $\mu <  0$ we arrive at the inequality
\beq
\label{mu01}
L_2[\psi_1] \geq L_3[\psi_1w^*] + L_4[\psi_1 u^*] \,,
\eeq
where $L_3[\phi]$ and $L_4[\phi]$ are defined in \eqref{eqfin4} and \eqref{eqfin5}. For a fixed constant $C$ in \eqref{eqfin5} and all sufficiently large $a$ we have $C\rho^{-2} \leq Ca^{-2}
< V_0$ on the support of $\psi_1$, which yields $L_4[\psi_1u^*] \geq 0 .$
Denote next $\psi_4 := \psi_1w^*\,\chi_{\varphi \in [-\pi, -\beta] \cup [\beta,\pi]}$, $\:\psi_5 := \psi_1w^*\,\chi_{\varphi \in  [-\beta, -\frac{\gamma}{2}]}$, and $\psi_6:= \psi_1w^*\,\chi_{\varphi \in  [\frac{\gamma}{2}, \beta ]}.$ Obviously, one has
$$
L_3[\psi_1w^*] \geq L_3[\psi_5] + L_3[\psi_6] +L_5[\psi_4]\,,
$$
where
$$
L_5[\phi] = \int \int \bigg\{ \Big|\frac{\partial \phi}{\partial \rho}\Big|^2 -\eps a^{-2}
\chi_{a \leq \rho \leq a(1+b)}\,|\phi|^2 \bigg\}\, \rho\, \D \rho\, \D \varphi
$$
and the integrals in $L_3$ and $L_5$ should be taken over the supports of the corresponding functions. Let us estimate $L_5[\psi_4]$. For fixed $\varphi$ we have
\begin{eqnarray}
\label{mu04}
\lefteqn{\int \bigg\{ \Big|\frac{\partial \phi}{\partial \rho}\Big|^2 -\eps a^{-2}\chi_{a \leq \rho \leq a(1+b)}\,|\phi|^2 \Big\}\, \rho\, \D \rho} \\[.3em] && \qquad \ge \int_a^{a(1+b)} \bigg\{ a\Big|\frac{\partial \phi}{\partial \rho}\Big|^2 -\eps a^{-2}a(1+b)\, |\phi|^2 \bigg\}\,\D \rho \nonumber \\[.3em] &&
\qquad \geq a \int_a^{a(1+b)}\bigg\{ \Big|\frac{\partial \phi}{\partial \rho}\Big|^2 -\eps a^{-2}(1+b)\,|\phi|^2 \bigg\}\,\D \rho
\geq 0 \nonumber
\end{eqnarray}
for any $\eps <\frac{\pi ^2}{b^2(1+b)}\,$, where in the last step we have used the inequality
$$
\int_0^d |f'|^2\, \D x \geq \frac{\pi^2}{d^2} \int _0^d |f|^2\, \D x\,,
$$
which holds for all $f \in \C^1(\R_+)$ with $f(0) = 0 .$ The inequality \eqref{mu04} yields $L_4[\psi_4] \geq 0$ for all sufficiently small $\eps$.

To complete the proof we need to show that $L_3[\psi_5] \geq 0$ and $L_3[\psi_6] \geq 0 .$ The estimates of  $L_3[\psi_5]$ repeat the procedure we used for the term $L_3[\psi_3]$ in the case $\mu < 0$ with the following modification. To assess the braced expression on the left-hand side of \eqref{eqfin7.1} for fixed a $x$ we cannot use Lemma~\ref{propChi}, because the function $\psi_5$ is not zero at $\varphi = \beta$, and consequently it is not in the in the domain of $h$. Instead, we apply Lemma \ref{propFunc} which yields
\begin{eqnarray*}
\lefteqn{ \int_0^{x\tan {(\beta-\frac{\gamma}{2})}} \Big\{\Big|\frac{\partial \psi_5}{\partial y}\Big|^2 
+V(y)|\psi_5|^2 \Big\}\, \D y - \alpha \sqrt{1+l'^2(x)}\,|\psi_5(x,l(x))|^2} \\ && \qquad \geq
 \int_{l(x)}^{x\tan {(\beta-\frac{\gamma}{2})}} \Big\{\Big|\frac{\partial \psi_5}{\partial y}\Big|^2 +
(1+l'^2(x))V_0|\psi_5|^2 \Big\}\, \D y \\ && \qquad - \alpha \sqrt{1+l'^2(x)}\,|\psi_5(x,l(x))|^2 -
V_0 l'^2(x)\int |\psi_5|^2 \D y \\ && \qquad \geq
-V_0 l'^2(x)\int |\psi_5|^2 \D y \geq -\int \frac{\epsilon}{|x|^2}|\psi_5|^2(x,y) \D y\,. \phantom{AAAAAAAA}
\end{eqnarray*}
Similar estimates can be used for $L_3[\psi_6]$.

\end{proof}
\subsection*{Acknowledgments}

\noindent The research was supported by the Czech Science Foundation (GA\v{C}R) within the project 14-06818S.


\end{document}